\documentclass[11pt,a4paper]{article}
\usepackage{fullpage}

\usepackage[utf8]{inputenc}
\usepackage{amsthm,amssymb}
\usepackage[leqno]{amsmath}
\usepackage{tcolorbox}
\usepackage{thmtools}
\usepackage{mathtools}
\usepackage{subcaption,thm-restate}
\usepackage[mathcal,mathscr]{eucal}
\usepackage{epsfig,graphicx,graphics,color}
\usepackage{caption}
\usepackage{enumerate}
\usepackage{enumitem}
\usepackage[sort]{cite}
\usepackage{hyperref}
\hypersetup{
    colorlinks=true,       
    linkcolor=blue,        
    citecolor=red,         
    filecolor=magenta,     
    urlcolor=cyan,         
    linktocpage=true
}

\theoremstyle{plain}
\newtheorem{thm}{Theorem}
\newtheorem{lem}[thm]{Lemma}

\newtheorem{claim}[thm]{Claim}

\theoremstyle{definition}
\newtheorem{rem}[thm]{Remark}

\makeatletter
\newcommand{\leqnomode}{\tagsleft@true}
\newcommand{\reqnomode}{\tagsleft@false}
\makeatother

\newcommand{\gp}{{G_\pi}}

\usepackage{algorithm}
\usepackage[noend]{algpseudocode}

\makeatletter
\renewcommand{\ALG@beginalgorithmic}{\small}
\makeatother

\usepackage{algorithm}
\usepackage{bbm}
\usepackage[noend]{algpseudocode}
\usepackage{varwidth}
\newcounter{algsubstate}


\def\final{0}  
\def\iflong{\iffalse}
\ifnum\final=0  
\newcommand{\knote}[1]{{\color{red}[{\tiny \textbf{Kristóf:} \bf #1}]}}
\newcommand{\ernote}[1]{{\color{blue}[{\tiny \textbf{Erika:} \bf #1}]}}
\newcommand{\evnote}[1]{{\color{orange}[{\tiny \textbf{Evelin:} \bf #1}]}}
\else 
\newcommand{\knote}[1]{}
\newcommand{\ernote}[1]{}
\newcommand{\evnote}[1]{}
\fi

\title{A dual approach for dynamic pricing\\ in multi-demand markets}

\author{
Krist\'of B\'erczi\thanks{MTA-ELTE Egerv\'ary Research Group, Department of Operations Research, E\"otvös Loránd University, Budapest, Hungary. Email: \texttt{kristof.berczi@ttk.elte.hu}.}
\and
Erika R. B\'erczi-Kov\'acs\thanks{MTA-ELTE Egerváry Research Group, Department of Operations Research, Eötvös Loránd University, Budapest, Hungary. Email: \texttt{erika.berczi-kovacs@ttk.elte.hu}.}
\and
Evelin Sz\"ogi\thanks{Department of Operations Research, Eötvös Loránd University, Budapest, Hungary. Email: \texttt{szogievelin@student.elte.hu}.}
}

\begin{document}
\maketitle

\begin{abstract}
Dynamic pricing schemes were introduced as an alternative to posted-price mechanisms. In contrast to static models, the dynamic setting allows to update the prices between buyer-arrivals based on the remaining sets of items and buyers, and so it is capable of maximizing social welfare without the need for a central coordinator. In this paper, we study the existence of optimal dynamic pricing schemes in combinatorial markets. In particular, we concentrate on multi-demand valuations, a natural extension of unit-demand valuations. The proposed approach is based on computing an optimal dual solution of the maximum social welfare problem with distinguished structural properties. 

Our contribution is twofold. By relying on an optimal dual solution, we show the existence of optimal dynamic prices in unit-demand markets and in multi-demand markets up to three buyers, thus giving new interpretations of results of Cohen-Addad et al. \cite{cohen2016invisible} and Berger et al. \cite{berger2020power}, respectively. Furthermore, we provide an optimal dynamic pricing scheme for bi-demand valuations with an arbitrary number of buyers. In all cases, our proofs also provide efficient algorithms for determining the optimal dynamic prices. 

\medskip

\end{abstract}

\section{Introduction}

A combinatorial market consists of a set of indivisible goods and a set of buyers, where each buyer has a valuation function that represents the buyer's preferences over the subsets of items. From an optimization point of view, the goal is to find an allocation of the items to buyers in such a way that the total sum of the buyers' values is maximized -- this sum is called the social welfare. An optimal allocation can be found efficiently in various settings \cite{nisan2006communication,clarke1971multipart,vickrey1961counterspeculation,groves1973incentives}, but the problem becomes significantly more difficult if one would like to realize the optimal social welfare through simple mechanisms.

A great amount of work concentrated on finding optimal pricing schemes. Given a price for each item, we define the utility of a buyer for a bundle of items to be the value of the bundle with respect to the buyer's valuation, minus the total price of the items in the bundle. A pair of pricing and allocation is called a Walrasian equilibrium if the market clears (that is, all the items are assigned to buyers) and everyone receives a bundle that maximizes her utility. Given any Walrasian equilibrium, the corresponding price vector is referred to as Walrasian pricing, and the definition implies that the corresponding allocation maximizes social welfare. 

Although Walrasian equilibria have distinguished properties, Cohen-Addad et al. \cite{cohen2016invisible} realized that the existence of a Walrasian equilibrium alone is not sufficient to achieve optimal social welfare based on buyers' decisions. Different bundles of items might have the same utility for the same buyer, and in such cases ties must be broken by a central coordinator in order to ensure that the optimal social welfare is achieved. However, the presence of such a tie-breaking rule is unrealistic in real life markets and buyers choose an arbitrary best bundle for themselves without caring about social optimum.

Dynamic pricing schemes were introduced as an alternative to posted-price mechanisms that are capable of maximizing social welfare even without a central tie-breaking coordinator. In this model, the buyers arrive in a sequential order, and each buyer selects a bundle of the remaining items that maximizes her utility. The buyers' preferences are known in advance, and the seller is allowed to update the prices between buyer arrivals based upon the remaining set of items, but without knowing the identity of the next buyer. The main open problem in~\cite{cohen2016invisible} asked whether any market with gross substitutes valuations has a dynamic pricing scheme that achieves optimal social welfare.

\paragraph{Related work.}

Walrasian equilibria were introduced already in the late 1800s \cite{walras1896elements} for divisible goods. A century later, Kelso and Crawford \cite{kelso1982job} defined gross substitutes functions and verified the existence of Walrasian prices for such valuations. It is worth mentioning that the class of gross substitutes functions coincides with that of M${}^\natural$-concave
 functions, introduced by Murota and Shioura~\cite{MurotaShioura1999}. The fundamental role of the gross substitutes condition was recognized by Gul and Stacchetti \cite{gul1999walrasian} who verified that it is necessary to ensure the existence of a Walrasian equilibrium. 

Cohen-Addad et al.~\cite{cohen2016invisible} and independently Hsu et al.~\cite{hsu2016prices} observed that Walrasian prices are not powerful enough to control the market on their own. The reason is that ties among different bundles must be broken in a coordinated fashion that is consistent with maximizing social welfare. Furthermore, this problem cannot be resolved by finding Walrasian prices where ties do not occur as \cite{hsu2016prices} showed that minimal Walrasian prices necessarily induce ties. To overcome these difficulties, \cite{cohen2016invisible} introduced the notion of dynamic pricing schemes, where prices can be redefined between buyer-arrivals. They proposed a scheme maximizing social welfare for matching or unit-demand markets, where the valuation of each buyer is determined by the most valuable item in her bundle. In each phase, the algorithm constructs a so-called `relation graph' and performs various computations upon it. Then the prices are updated based on structural properties of the graph.

Berger et al.~\cite{berger2020power} considered markets beyond unit-demand valuations, and provided a polynomial-time algorithm for finding optimal dynamic prices up to three multi-demand buyers. Their approach is based on a generalization of the relation graph of \cite{cohen2016invisible} that they call a `preference graph', and on a new directed graph termed the `item-equivalence graph'. They showed that there is a strong connection between these two graphs, and provided a pricing scheme based on these observations.

Further results on posted-price mechanisms considered matroid rank valuations \cite{berczi2020market}, relaxations such as combinatorial Walrasian equilibrium~\cite{feldman2016combinatorial}, and online settings \cite{blumrosen2008posted,chawla2010multi,chawla2010power,feldman2014combinatorial,dutting2016posted,dutting2017prophet,chawla2019pricing,ezra2018pricing,eden2019max}.  

\paragraph{Our contribution.}

In the present paper, we focus on multi-demand combinatorial markets. In this setting, each buyer $t$ has a positive integer bound $b(t)$ on the number of desired items, and the value of a set is the sum of the values of the $b(t)$ most valued items in the set. In particular, if we set each $b(t)$ to one then we get back the unit-demand case.

For multi-demand markets, the problem of finding an allocation that maximizes social welfare is equivalent to a maximum weight $b$-matching problem in a bipartite graph with vertex classes corresponding to the buyers and items, respectively. Note that, unlike in the case of Walrasian equilibrium, clearing the market is not required as a maximum weight $b$-matching might leave some of the items unallocated. The high level idea of our approach is to consider the dual of this problem, and to define an appropriate price vector based on an optimal dual solution with distinguished structural properties. 

Based on the primal-dual interpretation of the problem, first we give a simpler proof of a result of Cohen-Addad et al.~\cite{cohen2016invisible} on unit-demand valuations. Although this can be considered a special case of bi-demand markets, we discuss it separately as an illustration of our techniques.

When the total demand of the buyers exceeds the number of available items, ensuring the optimality of the final allocation becomes more intricate. Therefore, we consider instances satisfying the following property:
\leqnomode
\begin{equation}
\text{each buyer $t\in T$ receives exactly $b(t)$ items in every optimal allocation.} \tag{OPT} \label{eq:FAC}
\end{equation}
\reqnomode
While this is a restrictive assumption, it is a reasonable condition that holds for a wide range of applications, and also appeared in \cite{berger2020power} and recently in \cite{pashkovich2022two}. For example, if the total number of items is not less than the total demand of the buyers and the value of each item is strictly positive for each buyer, then it is not difficult to check that \eqref{eq:FAC} is satisfied. 

The problem becomes significantly more difficult for larger demands. Berger et al. \cite{berger2020power} observed that bundles that are given to a buyer in different optimal allocations satisfy strong structural properties. For markets up to three multi-demand buyers, they grouped the items into at most eight equivalence classes based on which buyer could get them in an optimal solution, and then analyzed the item-equivalence graph for obtaining an optimal dynamic pricing. We show that, when assumption \eqref{eq:FAC} is satisfied, these properties follow from the primal-dual interpretation of the problem, and give a new proof of their result for such instances.

The main result of the paper is an algorithm for determining optimal dynamic prices in bi-demand markets with an arbitrary number of buyers, that is, when the demand $b(t)$ is two for each buyer $t$.\footnote{In a recent manuscript, Pashkovich and Xie \cite{pashkovich2022two} showed that the result of Berger et al.~\cite{berger2020power} can be generalized from three to four buyers. They further extended the results of the current paper on bi-demand valuations to the case when each buyer is ready to buy up to three items.} Besides structural observations on the dual solution, the proof relies on uncrossing sets that are problematic in terms of resolving ties.
\medskip

The paper is organized as follows. Basic definitions and notation are given in Section~\ref{sec:prelim}, while Section~\ref{sec:bmatch} provides structural observations on optimal dynamic prices in multi-demand markets. Unit- and multi-demand markets up to three buyers are discussed in Section~\ref{sec:unitthree}. Finally, Section~\ref{sec:bi} solves the bi-demand case under the \eqref{eq:FAC} condition.

\section{Preliminaries}
\label{sec:prelim}

\paragraph{Basic notation.}

We denote the sets of \emph{real}, \emph{non-negative real}, \emph{integer}, and \emph{positive integer} numbers by $\mathbb{R}$, $\mathbb{R}_+$, $\mathbb{Z}$, and $\mathbb{Z}_{>0}$, respectively. Given a ground set $S$ and subsets $X,Y\subseteq S$, the \emph{difference} of $X$ and $Y$ is denoted by $X-Y$. If $Y$ consists of a single element $y$, then $X-\{y\}$ and $X\cup\{y\}$ are abbreviated by $X-y$ and $X+y$, respectively. The \emph{symmetric difference} of $X$ and $Y$ is $X\triangle Y\coloneqq (X-Y)\cup(Y-X)$. For a function $f\colon S\rightarrow\mathbb{R}$, the total sum of its values over a set $X$ is denoted by $f(X)\coloneqq \sum_{s\in X} f(s)$. The \emph{inner product} of two vectors $x, y \in \mathbb{R}^S$ is $x \cdot y \coloneqq  \sum_{s\in S} x(s) y(s)$. Given a set $S$, an \emph{ordering} of $S$ is a bijection $\sigma$ between $S$ and the set of integers $\{1,\dots,|S|\}$. For a set $X\subseteq S$, we denote the restriction of the ordering to $S-X$ by $\sigma|_{S-X}$. Given orderings $\sigma_1$ and $\sigma_2$ of disjoint sets $S_1$ and $S_2$, respectively, we denote by $\sigma=(\sigma_1,\sigma_2)$ the ordering of $S\coloneqq S_1\cup S_2$ where $\sigma(s)=\sigma_1(s)$ for $s\in S_1$ and $\sigma_2(s)+|S_1|$ for $s\in S_2$. 

Let $G=(S,T;E)$ be a bipartite graph with vertex classes $S$ and $T$ and edge set $E$. We will always denote the \emph{vertex set} of the graph by $V\coloneqq S\cup T$. For a subset $X\subseteq V$, we denote the \emph{set of edges induced by $X$} by $E[X]$, while $G[X]$ stands for the \emph{graph induced by $X$}. The graph obtained from $G$ by \emph{deleting $X$} is denoted by $G-X$. Given a subset $F\subseteq E$, the \emph{set of edges in $F$ incident to a vertex} $v\in V$ is denoted by $\delta_F(v)$. Accordingly, the \emph{degree} of $v$ in $F$ is $d_F(v)\coloneqq |\delta_F(v)|$. For a set $Z\subseteq T$, the \emph{set of neighbors of $Z$ with respect to $F$} is denoted by $N_F(Z)$, that is, $N_F(Z)\coloneqq \{s\in S\mid \text{there exists and edge $st\in F$ with $t\in Z$}\}$. The subscript $F$ is dropped from the notation or is changed to $G$ whenever $F$ is the whole edge set. 

\paragraph{Market model.}

A combinatorial market consists of a set $S$ of \emph{indivisible items} and a set $T$ of \emph{buyers}. We consider \emph{multi-demand}\footnote{Multi-demand valuations are special cases of weighted matroid rank functions for uniform matroids, see \cite{berczi2020market}.} markets, where each buyer $t\in T$ has a valuation $v_t\colon S\rightarrow\mathbb{R}_+$ over individual items together with an upper bound $b(t)$ on the number of desired items, and the \emph{value} of a set $X\subseteq S$ for buyer $t$ is defined as $v_t(X)\coloneqq \max\{v_t(X')\mid X'\subseteq X,|X'|\leq b(t)\}$. \emph{Unit-demand} and \emph{bi-demand} valuations correspond to the special cases when $b(t)=1$ and $b(t)=2$ for each $t\in T$, respectively.

Given a \emph{price vector} $p\colon S\rightarrow\mathbb{R}_+$, the \emph{utility} of buyer $t$ for $X$ is defined as $u_t(X)\coloneqq v_t(X)-p(X)$. The buyers, whose valuations are known in advance, arrive in an undetermined order, and the next buyer always chooses a subset of at most her desired number of items that maximizes her utility. In contrast to static models, the prices can be updated between buyer-arrivals based on the remaining sets of items and buyers. The goal is to set the prices at each phase in such a way that no matter in what order the buyers arrive, the final allocation maximizes the social welfare. Such a pricing scheme and allocation are called \emph{optimal}. It is worth emphasizing that a buyer may decide either to take or not to take an item which has $0$ utility, that is, it might happen that the bundle of items that she chooses is not inclusionwise minimal. This seemingly tiny degree of freedom actually results in difficulties that one has to take care of.  

\begin{lem}\label{lem:all}
We may assume that all items are allocated in every optimal allocation. 
\end{lem}
\begin{proof}
One can find an optimal allocation that uses an inclusionwise minimum number of items by relying on a weighted $b$-matching algorithm, see \cite{}. Setting the price of unused items to a large value ensures that no buyer takes them. Hence every optimal allocation uses the same set of items, meaning that the remaining items play no role in the problem and so can be deleted.
\end{proof}

In particular, when \eqref{eq:FAC} is assumed, Lemma~\ref{lem:all} implies that the number of items coincides with the total demand of the buyers.

\paragraph{Weighted $b$-matchings.}

Let $G=(S,T;E)$ be a bipartite graph and recall that $V\coloneqq S\cup T$. Given an upper bound $b\colon V\rightarrow\mathbb{Z}_+$ on the vertices, a subset $M\subseteq E$ is called a \emph{$b$-matching} if $d_M(v)\leq b(v)$ for every $v\in V$. If equality holds for each $v\in V$, then $M$ is called a \emph{$b$-factor}.  Notice that if $b(v)=1$ for each $v\in V$, then a $b$-matching or $b$-factor is simply a matching or perfect matching, respectively. K\H{o}nig's classical theorem~\cite{konig1916graphen} gives a necessary and sufficient condition for the existence of a perfect matching in a bipartite graph.

\begin{thm}[K\H{o}nig] \label{thm:konig}
There exists a perfect matching in a bipartite graph $G=(S,T;E)$ if and only if $|S|=|T|$ and $|N(Y)|\geq|Y|$ for every $Y\subseteq T$.
\end{thm}

Let $w:E\rightarrow\mathbb{R}$ be a weight function on the edges. A function $\pi\colon V\rightarrow\mathbb{R}$ on the vertex set $V=S\cup T$ is a \emph{weighted covering} of $w$ if $\pi(s)+\pi(t)\geq w(st)$ holds for every edge $st\in E$. An edge $st$ is called \emph{tight with respect to $\pi$} if $\pi(s)+\pi(t)=w(st)$.  
The \emph{total value} of the covering is $\pi\cdot b=\sum_{v\in V}\pi(v)\cdot b(v)$. We refer to a covering of minimum total value as \emph{optimal}. The celebrated result of Egerváry~\cite{egervary1931matrixok} provides a min-max characterization for the maximum weight of a matching or a perfect matching in a bipartite graph.

\begin{thm}[Egerv\'ary] \label{thm:egervary}
Let $G=(S,T;E)$ be a graph, $w:W\rightarrow\mathbb{R}$ be a weight function. Then the maximum weight of a matching is equal to the minimum total value of a non-negative weighted covering $\pi$ of $w$. If $G$ has a perfect matching, then the maximum weight of a perfect matching is equal to the minimum total value of a weighted covering $\pi$ of $w$.
\end{thm}


\section{Multi-demand markets and maximum weight \texorpdfstring{$b$}{b}-matchings}
\label{sec:bmatch}

A combinatorial market with multi-demand valuations can be naturally  identified with an edge-weighted complete bipartite graph $G=(S,T;E)$ where $S$ is the set of items, $T$ is the set of buyers, and for every item $s$ and buyer $t$ the weight of edge $st\in E$ is $w(st)\coloneqq v_t(s)$. We extend the demands to $S$ as well by setting $b(s)=1$ for every $s\in S$. Then an optimal allocation of the items corresponds to a maximum weight subset $M\subseteq E$ satisfying $d_M(v)\leq b(v)$ for each $v\in S\cup T$.

\subsection{Structure of weighted coverings}
\label{sec:covering}

In general, a $b$-factor or even a maximum weight $b$-matching can be found in polynomial time (even in non-bipartite graphs, see e.g. \cite{schrijver2003combinatorial}). When $b$ is identically one on $S$, then the following folklore characterization follows easily from K\H{o}nig's and Egerv\'ary's theorems\footnote{The same results follow by strong duality applied to the linear programming formulations of the problems.}. 

\begin{lem} \label{lem:bm}
Let $G=(S,T;E)$ be a bipartite graph, $w:E\rightarrow\mathbb{R}_+$ be a weight function, and $b\colon V\rightarrow\mathbb{Z}_{>0}$ be an upper bound function satisfying $b(s)=1$ for $s\in S$. 
\begin{enumerate}[label=(\alph*),topsep=2pt,itemsep=0pt,partopsep=0pt,parsep=1pt]
    \item $G$ has a $b$-factor if and only if $|S|=b(T)$ and $|N(X)|\geq b(X)$ for every $X\subseteq T$. \label{eq:a}
    \item The maximum $w$-weight of a $b$-matching is equal to the minimum total value of a non-negative weighted covering $\pi$ of $w$. \label{eq:b}
\end{enumerate}
\end{lem}
\begin{proof}
Let $G'=(S',T;E')$ denote the graph obtained from $G$ by taking $b(t)$ copies of each vertex $t\in T$ and connecting them to the vertices in $N_G(t)$. It is not difficult to check that $G$ has a $b$-factor if and only if $G'$ has a perfect matching, thus first part of the theorem follows by Theorem~\ref{thm:konig}.

To see the second part, for each copy $t'\in T'$ of an original vertex $t\in T$, define the weight of edge $st'$ as $w'(st')\coloneqq w(st)$. Then the maximum $w$-weight of a $b$-matching of $G$ is equal to the maximum $w'$-weight of a matching of $G'$. Now take an optimal non-negative weighted covering $\pi'$ of $w'$ in $G'$. As the different copies of an original vertex $t\in T$ share the same neighbors in $G'$, each of them receive the same value in any optimal weighted covering of $w'$ - define $\pi(t)$ to be this value. Then $\pi$ is a non-negative weighted covering of $w$ in $G$ with total value equal to that of $\pi'$, hence the theorem follows by Theorem~\ref{thm:egervary}.
\end{proof}

Given a weighted covering $\pi$, the \emph{subgraph of tight edges} with respect to $\pi$ is denoted by $G_\pi=(S,T;E_\pi)$. In what follows, we prove some easy structural results on the relation of optimal $b$-matchings and weighted coverings.

\begin{lem} \label{lem:opt}
Let $G=(S,T;E)$ be a bipartite graph, $w:E\rightarrow\mathbb{R}_+$ be a weight function, and $b\colon V\rightarrow\mathbb{Z}_{>0}$ be an upper bound function satisfying $b(s)=1$ for $s\in S$.
\begin{enumerate}[label=(\alph*),topsep=2pt,itemsep=0pt,partopsep=0pt,parsep=1pt]
    \item For any optimal non-negative weighted covering $\pi$ of $w$, a $b$-matching $M\subseteq E$ has maximum weight if and only if $M\subseteq E_\pi$ and $d_M(v)=b(v)$ for each $v$ with $\pi(v)>0$. \label{en:match}
    \item For any optimal weighted covering $\pi$ of $w$, a $b$-factor $M\subseteq E$ has maximum weight if and only if $M\subseteq E_\pi$. \label{en:fac}
\end{enumerate}
\end{lem}
\begin{proof}
Let $M$ be a maximum weight $b$-matching and $\pi$ be an optimal non-negative weighted covering. We have $w(M)=\sum_{st\in M}w(st)\leq\sum_{st\in M}(\pi(s)+\pi(t))\leq \sum_{v\in V} \pi(v)\cdot b(v)$, and equality holds throughout if and only if $M$ consists of tight edges and $\pi(v)=0$ if $d_M(v)<b(v)$.

Now consider the $b$-factor case. Let $M$ be a maximum weight $b$-factor and $\pi$ be an optimal weighted covering. We have $w(M)=\sum_{st\in M}w(st)\leq\sum_{st\in M}(\pi(s)+\pi(t))=\sum_{v\in V} \pi(v)\cdot b(v)$, and the inequality is satisfied with equality if and only if $M$ consists of tight edges.
\end{proof}

Following the notation of \cite{berger2020power}, we call an edge $st\in E$ \emph{legal} if there exists a maximum weight $b$-matching containing it, and say that $s$ is \emph{legal for $t$}. A subset $F\subseteq\delta(t)$ is \emph{feasible} if there exists a maximum weight $b$-matching $M$ such that $\delta_M(t)=F$; in this case $N_F(t)$ is called \emph{feasible for $t$}\footnote{The notion of feasibility is closely related to `legal allocations' introduced in \cite{berger2020power}. However, `legal subsets' are different from feasible ones, hence we use a different term here to avoid confusion.}.  Notice that a feasible set necessarily consists of legal edges. The essence of the following technical lemma is that there exists an optimal non-negative weighted covering for which $\gp$ consists only of legal edges, thus giving a better structural understanding of optimal dual solutions; for an illustration see Figure~\ref{fig:ex}.

\begin{lem} \label{lem:dual}
The optimal $\pi$ attaining the minimum in Lemma~\ref{lem:bm}\ref{eq:b} can be chosen such that
\begin{enumerate}[label=(\alph*),topsep=2pt,itemsep=0pt,partopsep=0pt,parsep=1pt]
    \item an edge $st$ is tight with respect to $\pi$ if and only if it is legal, and \label{prop:a}
    \item $\pi(v)=0$ for some $v\in V$ if and only if there exists a maximum weight $b$-matching $M$ with $d_M(v)<b(v)$. \label{prop:b}
\end{enumerate}
Furthermore, such a $\pi$ can be determined in polynomial time.
\end{lem}
\begin{proof}
In both cases, the `if' part follows by Lemma~\ref{lem:opt}. Let $M$ and $\pi$ be a maximum weight $b$-matching and an optimal non-negative weighted covering, respectively. To prove the lemma, we will modify $\pi$ in two phases. 

In the first phase, we ensure \ref{prop:a} to hold. Take an arbitrary ordering $e_1,\dots,e_m$ of the edges, and set $\pi_0\coloneqq \pi$ and $w_0\coloneqq w$. For $i=1,\dots,m$, repeat the following steps. Let $\varepsilon_i\coloneqq \max\{w_{i-1}(M)\mid \text{$M$ is a $b$-matching}\}-\max\{w_{i-1}(M)\mid\text{$M$ is a $b$-matching containing $e_i$}\}$. Notice that $\varepsilon_i>0$ exactly if $e_i$ is not legal. Let $w_i$ denote the weight function obtained from $w_{i-1}$ by increasing the weight of $e_i$ by $\varepsilon_i/2$, and let $\pi_i$ be an optimal non-negative weighted covering of $w_i$. Due to the definition of $\varepsilon_i$, a $b$-matching $M$ has maximum weight with respect to $w_i$ if and only if it has maximum weight with respect to $w_{i-1}$, and in this case $w_{i}(M)=w_{i-1}(M)$. That is, the sets of maximum weight $b$-matchings with respect to $w$ and $w_m$ coincide, and the weights of legal edges do not change, therefore $\pi_m$ is an optimal non-negative weighted covering of $w$ as well.

In the second phase, we ensure \ref{prop:b} to hold. Take an arbitrary ordering $v_1,\dots,v_n$ of the vertices, 
for $j=1,\dots,n$, repeat the following steps. Let $\delta_j\coloneqq \max\{w_{m+j-1}(M)\mid\text{$M$ is a $b$-matching}\}-\max\{w_{m+j-1}(M)\mid\text{$M$ is a $b$-matching, $d_M(v_j)\leq b(v_j)-1$}\}$. Then $\delta_j>0$ if and only if the degree of $v_j$ is $b(v_j)$ in every maximum weight $b$-matching. Let $w_{m+j}$ denote the weight function obtained from $w_{m+j-1}$ by decreasing the weight of the edges incident to $v_j$ by $\delta_j/(b(v_j)+1)$ and let $\pi_{m+j}$ be an optimal non-negative weighted covering of $w_{m+j}$. Due to the definition of $\delta_j$, a $b$-matching $M$ has maximum weight with respect to $w_{m+j-1}$ if and only if it has maximum weight with respect to $w_{m+j}$, and in this case $w_{m+j}(M)=w_{m+j-1}(M)-\delta_j\cdot b(v_j)$. That is, the sets of maximum weight $b$-matchings with respect to $w$ and $w_{m+n}$ coincide. Let $\pi'$ denote the weighted covering of $w$ obtained by increasing the value of $\pi_{m+n}(v_\ell)$ by $\delta_\ell/(b(v_\ell)+1)$ for $\ell=1,\dots,n$. As the total value of $\pi'$ is greater than that of $\pi_{m+n}$ by exactly $\max\{w(M)\mid\text{$M$ is a $b$-matching}\}-\max\{w_{m+n}(M)\mid\text{$M$ is a $b$-matching}\}$, $\pi'$ is an optimal non-negative weighted covering of $w$.

As $\varepsilon_i>0$ whenever $e_i$ is not legal and $\delta_j>0$ whenever there is no a maximum weight $b$-matching $M$ with $d_M(v_j)<b(v_j)$, $\pi'$ satisfies both \ref{prop:a} and \ref{prop:b} as required.
\end{proof}

\begin{rem}
If the market satisfies property \eqref{eq:FAC}, the lemma implies that there exists an optimal non-negative weighted covering that is positive for every buyer and every item.
\end{rem}

Feasible sets play a key role in the design of optimal dynamic pricing schemes. After the current buyer leaves, the associated bipartite graph is updated by deleting the vertices corresponding to the buyer and her bundle of items, and the prices are recomputed for the remaining items. It follows from the definitions that the scheme is optimal if and only if the prices are always set in such a way that any bundle of items maximizing the utility of an agent $t$ forms a feasible set for $t$. 

\subsection{Adequate orderings}
\label{sec:adequate}

The high-level idea of our approach is as follows. First, we take an optimal non-negative weighted covering $\pi$ provided by Lemma~\ref{lem:dual}. 
If we define the price of an item $s\in S$ to be $\pi(s)$, then for any $t\in T$ we have $u_t(s)=v_t(s)-\pi(s)=w(st)-\pi(s)\leq\pi(t)$ and, by Lemma~\ref{lem:dual}\ref{prop:a}, equality holds if and only if $s$ is feasible for $t$. This means that each buyer prefers choosing items that are legal for her. For unit-demand valuations, such a solution immediately yields an optimal dynamic pricing scheme as explained in Section~\ref{sec:unit}. However, when the demands are greater than one, a collection of legal items might not form a feasible set, see an example on Figure~\ref{fig:ex}. In order to control the choices of the buyers, we slightly perturb the item prices by choosing an ordering $\sigma\colon S\to\{1,\dots,|S|\}$ and set the price of item $s$ to be $\pi(s)+\delta\cdot\sigma(s)$ for some sufficiently small $\delta>0$. Here the value of $\sigma(s)$ will be set in such a way that any bundle of items maximizing the utility of a buyer will form a feasible set for her, as needed. 

\begin{figure}[t!]
\centering
\begin{subfigure}[t]{0.45\textwidth}
  \centering
  \includegraphics[width=.8\linewidth]{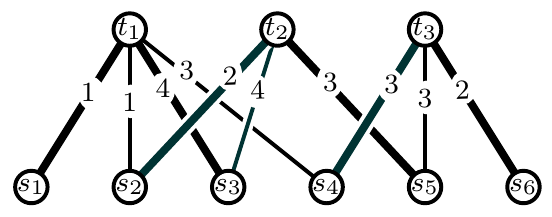}
  \caption{Maximum weight $b$-matching $M_1=\{t_1s_1,t_1s_3,t_2s_2,t_2s_5,t_3s_4,t_3s_6\}$.}
  \label{fig:ex1}
\end{subfigure}\hfill
\begin{subfigure}[t]{0.45\textwidth}
  \centering
  \includegraphics[width=.8\linewidth]{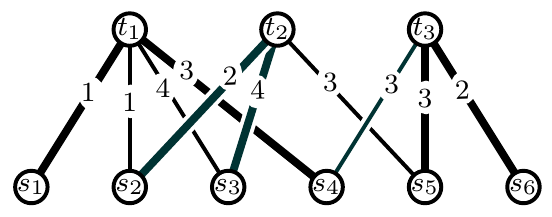}
  \caption{Maximum weight $b$-matching $M_2=\{t_1s_1,t_1s_4,t_2s_2,t_2s_3,t_3s_5,t_3s_6\}$}
  \label{fig:ex2}
\end{subfigure}\vspace{5pt}
\begin{subfigure}[t]{0.45\textwidth}
  \centering
  \includegraphics[width=.8\linewidth]{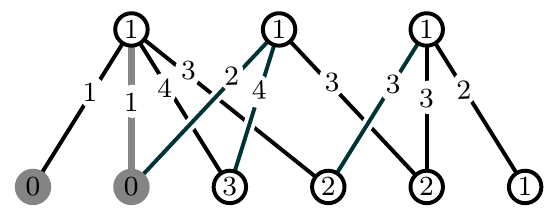}
  \caption{An optimal non-negative weighted covering $\pi$.  Notice that $s_1t_1$ is tight but not legal, and $\pi(s_1)=\pi(s_2)=0$ although $d_M(s_1)=d_M(s_2)=1$ for every maximum weight $b$-matching.}
  \label{fig:ex3}
\end{subfigure}\hfill
\begin{subfigure}[t]{0.45\textwidth}
  \centering
  \includegraphics[width=.8\linewidth]{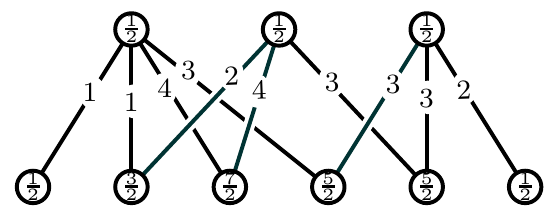}
  \caption{An optimal non-negative weighted covering satisfying the conditions of Lemma~\ref{lem:dual}.}
  \label{fig:ex4}
\end{subfigure}
\caption{A bipartite graph corresponding to a market with three buyers having demand two and six items. The numbers denote the weights of the edges; all the remaining edges have weight $0$. There are two maximum weight $b$-matchings $M_1$ (Figure~\ref{fig:ex1}) and $M_2$ (Figure~\ref{fig:ex2}).  Notice that both $s_3t_1$ and $s_4t_1$ are legal, but they do not form a feasible set.}
\label{fig:ex}
\end{figure}

Given a bipartite graph $G=(S,T;E)$ and upper bounds $b\colon V\to\mathbb{Z}_{>0}$ with $b(s)=1$ for $s\in S$, we call an ordering $\sigma\colon S\to\{1,\dots,|S|\}$ \emph{adequate for $G$} if it satisfies the following condition: for any $t\in T$, there exists a $b$-factor in $G$ that matches $t$ to its first $b(t)$ neighbors according to the ordering $\sigma$. 
For ease of notation, we introduce the \emph{slack} of $\pi$ to denote $\Delta(\pi)\coloneqq \min\bigl\{\min\{\pi(t)+\pi(s)-w(st)\mid st\in E,st\ \text{is not tight}\},\ \min\{\pi(v)\mid v\in V,\pi(v)>0\}\bigr\}$, where the minimum over an empty set is defined to be $+\infty$. Using this terminology, the above idea is formalized in the following lemma.

\begin{lem}\label{lem:fac}
Assume that \eqref{eq:FAC} is satisfied. Let $G=(S,T;E)$ be the edge-weighted bipartite graph associated with the market, $\pi$ be a weighted covering provided by Lemma~\ref{lem:dual}, and $\sigma$ be an adequate ordering for $G_\pi$. For $\delta\coloneqq \Delta(\pi)/(|S|+1)$, setting the prices to $p(s)\coloneqq \pi(s)+\delta\cdot\sigma(s)$ results in optimal dynamic prices.
\end{lem}
\begin{proof}
By \eqref{eq:FAC}, every optimal solution is a $b$-factor. Observe that for any $s\in S$ and $t\in T$, we have 
\begin{align*}
u_t(s)
{}&{}=
v_t(s)-p(s)\\
{}&{}=
w(st)-(\pi(s)+\delta\cdot\sigma(s))\\
{}&{}\leq \pi(t)-\delta\cdot\sigma(s).
\end{align*}
Here equality holds if and only if $st$ is tight with respect to $\pi$, in which case $u_t(s)=\pi(t)-\delta\cdot\sigma(s)>\pi(t)-\Delta(\pi)\cdot |S|/(|S|+1)>0$ by the choice of $\delta$ and by Lemma~\ref{lem:dual}\ref{prop:b}. Furthermore, if $st$ is tight and  $s't$ is a non-tight edge of $G$, then $u_t(s')\leq\pi(t)-\Delta(\pi)\leq\pi(t)-\delta(|S|+1)<u_t(s)$ by the choice of $\delta$. 
Concluding the above, we get that no matter which buyer arrives next, she strictly prefers legal items over non-legal ones, and legal items have strictly positive utility values for her. That is, she chooses the first $b(t)$ of its neighbors in $G_\pi$ according to the ordering $\sigma$. As $\sigma$ is adequate for $G_\pi$, the statement follows by Lemma~\ref{lem:opt}\ref{en:fac}.
\end{proof}

It is worth emphasizing that the application of Lemma~\ref{lem:fac} provides optimal dynamic prices for a single round; the prices should be updated before the arrival of each buyer accordingly.

For a $\pi\colon V\to\mathbb{R}$ and $\sigma\colon S\to\{1,\dots,|S|\}$, the \emph{combination} of $\pi$ and $\sigma$ is an ordering $\sigma'\colon S\to\{1,\dots,|S|\}$ that is obtained by pre-ordering the elements of $S$ according to their $\pi$ values in a non-decreasing order, and then items having the same $\pi$ value are ordered according to $\sigma$. We denote the combination of $\pi$ and $\sigma$ by $\pi\circ\sigma$. The following technical lemma will be useful in the inductive proof.

\begin{lem}\label{lem:fix}
Let $G=(S,T;E)$ be an edge-weighted bipartite graph with all edges having weight one, and $b\colon V\to\mathbb{Z}_{>0}$ be an upper bound function satisfying $b(s)=1$ for $s\in S$ such that $G$ admits a $b$-factor. Furthermore, let $\pi$ be a weighted covering provided by Lemma~\ref{lem:dual}, and $\sigma$ be an adequate ordering for $G_\pi=(S,T;E_\pi)$. Then $\pi\circ\sigma$ is an adequate ordering for $G$.
\end{lem}
\begin{proof}
Let $\sigma'\coloneqq\pi\circ\sigma$ denote the combination of $\pi$ and $\sigma$. We claim that for any $t\in T$, the first $b(t)$ neighbors of $t$ in $G_\pi$ according to $\sigma$ coincides with the first $b(t)$ neighbors of $t$ in $G$ according to $\sigma'$. Indeed, this follows from the fact that the edge-weights are identically $1$, hence the value of $\pi(s)$ is exactly $1-\pi(t)$ if $st\in E_\pi$ and strictly less if $st\in E\setminus E_\pi$. That is, in the ordering $\sigma'$, the edges in $E_\pi$ precede the edges in $E\setminus E_\pi$. As $G$ admits a $b$-factor by assumption, $t$ has at least $b(t)$ neighbors in $G_\pi$, and the lemma follows.
\end{proof}

\section{Unit- and multi-demand markets}
\label{sec:unitthree}

\subsection{Unit-demand markets}
\label{sec:unit}

Based on the primal-dual interpretation of the problem, first we give a simpler proof of a result of Cohen-Addad et al.~\cite{cohen2016invisible} on unit-demand valuations as an illustration of our approach.

\begin{thm}[Cohen-Addad et al.]
\label{thm:unit}
Every unit-demand market admits an optimal dynamic pricing that can be computed in polynomial time.
\end{thm}
\begin{proof}
Consider the bipartite graph associated with the market, take an optimal cover $\pi$ provided by Lemma~\ref{lem:dual}, and 
set the price of item $s$ to be $\pi(s)$.
For a pair of buyer $t\in T$ and $s\in S$, we have 
\begin{align*}
u_t(s)
{}&{}=
v_t(s)-p(s)\\
{}&{}=
w(st)-p(s)\\
{}&{}\leq 
(\pi(s)+\pi(t))-\pi(s)\\
{}&{}
=\pi(t).
\end{align*}
By 
Lemma~\ref{lem:dual}\ref{prop:a}, strict equality holds if and only if $st$ is legal. We claim that no matter which buyer arrives next, she either chooses an item that is legal (and so forms a feasible set for her), or she takes none of the items and the empty set is feasible for her. 

To see this, assume first that $\pi(t)>0$. By Lemma~\ref{lem:dual}\ref{prop:b}, there exists at least one item legal for $t$, and those items are exactly the ones maximizing her utility. Now assume that $\pi(t)=0$. By Lemma~\ref{lem:dual}\ref{prop:b}, the empty set is feasible for $t$. Furthermore, for any item $s\in S$, the utility $u_t(s)$ is negative unless $s$ is legal for $t$, in which case $u_t(s)=0$.  Notice that a buyer may decide to take or not to take any item with zero utility value. However, she gets a feasible set in both cases by the above, thus concluding the proof.
\end{proof}

\subsection{Multi-demand markets up to three buyers}
\label{sec:three}

The aim of the section is to settle the existence of optimal dynamic prices in multi-demand markets with a bounded number of buyers, under the assumption \eqref{eq:FAC}.

\begin{thm}[Berger et al.]
\label{thm:multithree}
Every multi-demand market with property \eqref{eq:FAC} and at most three buyers admits an optimal dynamic pricing scheme, and such prices can be computed in polynomial time.
\end{thm}

\begin{proof}
By Lemma~\ref{lem:fac}, it suffices to show the existence of an adequate ordering for $G_\pi$, where $\pi$ is an optimal non-negative weighted covering provided by Lemma~\ref{lem:dual}. For a single buyer, the statement is meaningless. For two buyers $t_1$ and $t_2$, $|S|=b(t_1)+b(t_2)$ by assumption \eqref{eq:FAC}. Let $\sigma$ be an ordering that starts with items in $N_\gp(t_1)\triangle N_\gp(t_2)$ and then puts the items in $N_\gp(t_1)\cap N_\gp(t_2)$ at the end of the ordering. Then, after the deletion of the first $b(t_i)$ neighbors of $t_i$ according to $\sigma$, the remaining $b(t_{3-i})$ items are in $N_\gp(t_{3-i})$, hence $\sigma$ is adequate. 

Now we turn to the case of three buyers. Let $t_1,t_2$ and $t_3$ denote the buyers, and let $b_i$, $v_i$, and $u_i$ denote the demand, valuation, and utility function corresponding to buyer $t_i$, respectively. Without loss of generality, we may assume that $b_1\geq b_2\geq b_3$. The proof is based on the observation that a set is feasible if and only if its deletion leaves `enough' items for the remaining buyers, formalized as follows.

\begin{claim} \label{cl:labeling}
A set $F\subseteq N_\gp(t_i)$ is feasible for $t_i$ if and only if $|F|=b_i$ and $|N_\gp(t_j)-F|\geq b_j$ for $j\neq i$. 
\end{claim}
\begin{proof}
The conditions are clearly necessary. To prove sufficiency, we show that the constraints of Lemma~\ref{lem:bm}\ref{eq:a} are fulfilled after deleting $t_i$ and $F$ from $\gp$, that is, $|S-F|=b(T)-b_i$ and $|N_\gp(Y)-F|\geq b(Y)$ for $Y\subseteq T-t_i$. By \eqref{eq:FAC} and the assumption that every item is legal for at least two buyers, $|S-F|=b(T)-b_i$ holds for $Y=T-t_i$. Furthermore, one-element subsets have enough neighbors by assumption, and the claim follows.
\end{proof}

For $I\subseteq\{1,2,3\}$, let $X_I\subseteq S$ denote the set of items that are legal exactly for buyers with indices in $I$, that is, $X_I\coloneqq \bigl(\bigcap_{i\in I} N_\gp(t_i)\bigr)-\bigl(\bigcup_{i\notin I} N_\gp(t_i)\bigr)$. We may assume that $X_1=X_2=X_3=\emptyset$. Indeed, given an adequate ordering for $\gp-(X_1\cup X_2\cup X_3)$ where the demands of $t_i$ is changed to $b_i-|X_i|$ for $i\in\{1,2,3\}$, putting the items in $X_1\cup X_2\cup X_3$ at the beginning of the ordering results in an adequate solution for the original instance. 

By assumption, $|X_{12}|+|X_{13}|+|X_{23}|+|X_{123}|=b_1+b_2+b_3$. Furthermore, $|X_{ij}|\leq b_i+b_j$ holds for $i\neq j$, as otherwise in any allocation there exists an item that is legal only for $t_i$ and $t_j$ but is not allocated to any of them, contradicting \eqref{eq:FAC}. We first define a labeling $\Theta\colon S\to\{1,2,3,4,5\}$ so that for each buyer $i$ and set $X_{ij}$, the number of items in $X_{ij}$ with label at most $4-i$ is $\max\{0, |X_{ij}|-b_j\}$. We will make sure that each buyer $i$ selects all items with label at most $4-i$ that are legal for her, which will be the key to satisfy the constraints of Claim \ref{cl:labeling}, see Figure~\ref{fig:three}.

\begin{figure}
    \centering
    \includegraphics[width=0.3\textwidth]{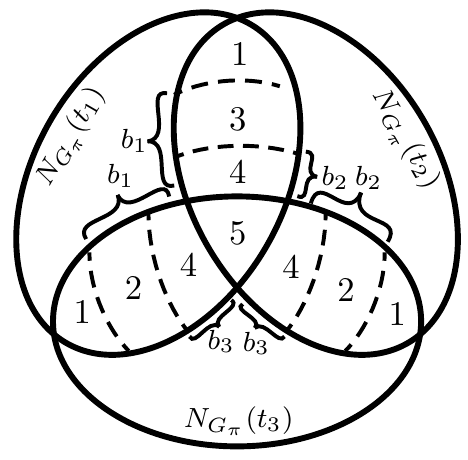}
    \caption{Definition of the labeling $\Theta$ for three buyers.  Notice that some parts might be empty, e.g. if $|X_{12}|\leq b_2$, then there are no items with label $1$ and $3$ in the intersection of $N_\gp(t_1)$ and $N_\gp(t_2)$.}
    \label{fig:three}
\end{figure}

All the items in $X_{123}$ are labeled by $5$. If $|X_{12}|\leq b_2$, then all the items in $X_{12}$ are labeled by $4$. If $b_1\geq |X_{12}|> b_2$, then $b_2$ items are labeled by $4$ and the remaining $|X_{12}|-b_2$ items are labeled by $3$ in $X_{12}$. If $|X_{12}|> b_1$, $b_2$ items are labeled by $4$, $b_1-b_2$ items are labeled by $3$, and the remaining $|X_{12}|-b_1$ items are labeled by $1$ in $X_{12}$. We proceed with $X_{13}$ analogously. If $|X_{13}|\leq b_3$, then all the items in $X_{13}$ are labeled by $4$. If $b_1\geq|X_{13}|> b_3$, then $b_3$ items are labeled by $4$ and the remaining $|X_{13}|-b_3$ items are labeled by $2$ in $X_{13}$. If $|X_{13}|>b_1$, $b_3$ items are labeled by $4$, $b_1-b_3$ items are labeled by $2$, and the remaining $|X_{13}|-b_1$ items are labeled by $1$ in $X_{13}$. Similarly, if $|X_{23}|\leq b_3$, then all the items in $X_{23}$ are labeled by $4$. If $b_2\geq |X_{23}|> b_3$, then $b_3$ items are labeled by $4$ and the remaining $|X_{23}|-b_3$ items are labeled by $2$ in $X_{23}$. If $|X_{23}|>b_2$, then $b_3$ items are labeled by $4$, $b_2-b_3$ items are labeled by $2$, and the remaining $|X_{23}|-b_2$ items are labeled by $1$ in $X_{23}$. 

Now let $\sigma$ be any ordering of the items satisfying the following condition: if the label of item $s_1$ is strictly less than that of item $s_2$, then $s_1$ precedes $s_2$ in the ordering, that is, $\Theta(s_1)<\Theta(s_2)$ implies $\sigma(s_1)<\sigma(s_2)$. We claim that $\sigma$ is adequate for $\gp$. To see this, it suffices to verify that the set $F$ of the first $b(t_i)$  neighbors of $t_i$ according to $\sigma$ fulfills the requirements of Claim~\ref{cl:labeling} for $i=1,2,3$. Let $\{i,j,k\}=\{1,2,3\}$. First we show that $F$ contains all the items $s\in X_{ij}\cup X_{ik}$ with $\Theta(s)\leq 4-i$. 

\begin{claim}\label{cl:4i}
We have $|\{s\in X_{ij}\cup X_{ik}\mid \Theta(s)\leq 4-i\}|\leq b_i$.
\end{claim}
\begin{proof}
Suppose to the contrary that this does not hold. Then $b_i<\max\{0,|X_{ij}|-b_j\}+\max\{0,|X_{ik}|-b_k\}$ by the definition of the labeling. 
Since $|X_{ij}|\leq b_i+b_j$ and $|X_{ik}|\leq b_i+b_k$, we have $\max\{0,|X_{ij}|-b_j\}\leq b_i$ and $\max\{0,|X_{ik}|-b_k\}\leq b_i$. Therefore, if
$b_i< \max\{0,|X_{ij}|-b_j\}+\max\{0,|X_{ik}|-b_k\}$, then both maximums must be positive on the right hand side. However, this leads to $b_i+b_j+b_k<|X_{ij}|+|X_{ik}|$, contradicting $b_i+b_j+b_k=|X_{ij}|+|X_{ik}|+|X_{jk}|+|X_{ijk}|$. 
\end{proof}

By Claim~\ref{cl:4i}, $F$ contains all the items $s\in X_{ij}\cup X_{ik}$ with $\Theta(s)\leq 4-i$, we have $|X_{ij}-F|\leq b_j$ and $|X_{ik}-F|\leq b_k$. Thus we get
\begin{align*}
    |N_\gp(t_j)-F|
    {}&{}=
    |X_{ij}-F|+|X_{jk}|+|X_{ijk}-F|\\
    {}&{}=
    |S|-|X_{ik}-F|-|F|\\
    {}&{}\geq
    (b_i+b_j+b_k)-b_k-b_i\\
    {}&{}=b_j.
\end{align*}
An analogous computation shows that $|N_\gp(t_k)-F|\geq b_k$. That is, $F$ is indeed a feasible set for $t_i$, concluding the proof of the theorem.
\end{proof}

\section{Bi-demand markets}
\label{sec:bi}

This section is devoted to the proof of the main result of the paper, the existence of optimal dynamic prices in bi-demand markets. The algorithms aims at identifying subsets of buyers whose neighboring set in $G_\pi$ is `small', meaning that other buyers should take no or at most one item from it. If no such set exists, then an adequate ordering is easy to find. Otherwise, by examining the structure of problematic sets, the problem is reduced to smaller instances. 

\begin{thm}
\label{thm:bi}
Every bi-demand market with property \eqref{eq:FAC} admits an optimal dynamic pricing scheme, and such prices can be computed in polynomial time.
\end{thm}

\begin{proof}
Let $G=(S,T;E)$ and $w$ be the bipartite graph and weight function associated with the market. Take an optimal non-negative weighted covering  $\pi$ of $w$ provided by Lemma~\ref{lem:dual}, and consider the subgraph $G_\pi=(S,T;E_\pi)$ of tight edges. For simplicity, we call a subset $M\subseteq E_\pi$ a \emph{$(1,2)$-factor} if $d_M(s)=1$ for every $s\in S$ and $d_M(t)=2$ for every $t\in T$. By \eqref{eq:FAC} and Lemmas~\ref{lem:all} and ~\ref{lem:opt}, there is a one-to-one correspondence between optimal allocations and $(1,2)$-factors of $G_\pi$. Therefore, by Lemma~\ref{lem:fac}, it suffices to show the existence of an adequate ordering $\sigma$ for $G_\pi$.

We prove by induction on $|T|$. The statement clearly holds when $|T|=1$, hence we assume that $|T|\geq 2$. As there exists a $(1,2)$-factor in $\gp$, we have $|N_\gp(Y)|\geq 2|Y|$ for every $Y\subseteq T$ by Lemma~\ref{lem:bm}\ref{eq:a}. We distinguish three cases. 
\medskip

\noindent \textbf{Case 1.} $|N_\gp(Y)|\geq 2|Y|+2$ for every $\emptyset\neq Y\subsetneq T$.

For any $t\in T$ and $s_1,s_2\in N_\gp(t)$, the graph $G_\pi-\{s_1,s_2,t\}$ still satisfies the conditions of Lemma~\ref{lem:bm}\ref{eq:a}, hence $\{s_1,s_2\}$ is feasible for $t$. Therefore, $\sigma$ can be chosen arbitrarily.
\medskip

\noindent \textbf{Case 2.} $|N_\gp(Y)|\geq 2|Y|+1$ for $\emptyset\neq Y\subsetneq T$ and there exists $Y$ for which equality holds.

We call a set $Y\subseteq T$ \emph{dangerous} if $|N_\gp(Y)|=2|Y|+1$. By Lemma~\ref{lem:bm}\ref{eq:a}, a pair $\{s_1,s_2\}\subseteq N_\gp(t)$ is not feasible for buyer $t$ if and only if there exists a dangerous set $Y\subseteq T-t$ with $s_1,s_2\in N_\gp(Y)$. In such case, we say that $Y$ \emph{belongs to buyer $t$}.  Notice that the same dangerous set might belong to several buyers.

\begin{claim}\label{cl:dang}
Assume that $Y_1$ and $Y_2$ are dangerous sets with $Y_1\cup Y_2\subsetneq T$.
\begin{enumerate}[label=(\alph*),topsep=2pt,itemsep=0pt,partopsep=0pt,parsep=1pt]
    \item If $Y_1\cap Y_2=\emptyset$ and $N_\gp(Y_1)\cap N_\gp(Y_2)\neq\emptyset$, then $|N_\gp(Y_1)\cap N_\gp(Y_2)|=1$ and $Y_1\cup Y_2$ is dangerous. \label{eq:cup}
    \item If $Y_1\cap Y_2\neq\emptyset$, then both $Y_1\cap Y_2$ and $Y_1\cup Y_2$ are dangerous. \label{eq:cap}
\end{enumerate}
\end{claim}
\begin{proof}
Observe that 
\begin{align*}
    (2|Y_1|+1)+(2|Y_2|+1)
    {}&{}=
    |N_\gp(Y_1)|+|N_\gp(Y_2)|\\
    {}&{}=
    |N_\gp(Y_1)\cap N_\gp(Y_2)|+|N_\gp(Y_1)\cup N_\gp(Y_2)|\\
    {}&{}=
    |N_\gp(Y_1)\cap N_\gp(Y_2)|+|N_\gp(Y_1\cup Y_2)|.
\end{align*}

Assume first that $Y_1\cap Y_2=\emptyset$. Then $|N_\gp(Y_1)\cap N_\gp(Y_2)|\leq 1$ as otherwise $|N_\gp(Y_1\cup Y_2)|\leq 2(|Y_1|+|Y_2|)=2|Y_1\cup Y_2|$, contradicting the assumption of Case 2. If $|N_\gp(Y_1)\cap N_\gp(Y_2)|=1$, then $|N_\gp(Y_1\cup Y_2)|=2|Y_1\cup Y_2|+1$ and so $Y_1\cup Y_2$ is dangerous.

Now consider the case when $Y_1\cap Y_2\neq\emptyset$. Then 
\begin{align*}
    |N_\gp(Y_1)\cap N_\gp(Y_2)|+|N_\gp(Y_1\cup Y_2)|
    {}&{}\geq 
    |N_\gp(Y_1\cap Y_2)|+|N_\gp(Y_1\cup Y_2)|\\
    {}&{}\geq
    (2|Y_1\cap Y_2|+1)+(2|Y_1\cup Y_2|+1)\\
    {}&{}=
    (2|Y_1|+1)+(2|Y_2|+1).
\end{align*}
Therefore, we have equality throughout, implying that both $Y_1\cap Y_2$ and $Y_1\cup Y_2$ are dangerous.
\end{proof}

Let $Z\subseteq T$ be an inclusionwise maximal dangerous set.
\medskip

\noindent \textbf{Subcase 2.1.} There is no dangerous set disjoint from $Z$.

First we show that if a pair $s_1,s_2\in N_\gp(t)$ is not feasible for a buyer $t\in T-Z$, then $s_1,s_2\in N_\gp(Z)$. Indeed, if $\{s_1,s_2\}$ is not feasible for $t$, then there is a dangerous set $X$ belonging to $t$ with $s_1,s_2\in N_\gp(X)$. Since $t\notin X\cup Z$ and $Z\cap X\neq\emptyset$ by the assumption of the subcase, Claim~\ref{cl:dang}\ref{eq:cap} applies implying that $X\cup Z$ is dangerous as well. The maximal choice of $Z$ implies $X\cup Z=Z$, hence $Z$ belongs to $t$ and $s_1,s_2\in N_\gp(Z)$.

Now take an arbitrary buyer $t_0\in T-Z$ who shares a neighbor with $Z$, and let $s_0\in N_\gp(t_0)\cap N_\gp(Z)$. Let $\sigma'$ be an arbitrary ordering of the items in $S-N_\gp(Z)$. Furthermore, let $G''\coloneqq \gp[Z\cup N_\gp(Z)]-s_0$. For any $(1,2)$-factor of $\gp$ containing $s_0t_0$, its restriction to $G''$ is a $(1,2)$-factor as well. In $G''$, some of the edges might not be contained in any of the $(1,2)$-factors. Still, by induction and Lemma~\ref{lem:fix}, there exists an adequate ordering $\sigma''$ of the items in $G''$. Finally, let $\sigma'''$ denote the trivial ordering of the single element set $\{s_0\}$. We set $\sigma\coloneqq (\sigma',\sigma'',\sigma''')$. Then any buyer $t\in T-Z$ will choose at most one item from $N_\gp(Z)$, hence the adequateness of $\sigma$ follows from that of $\sigma''$ and the assumption of the subcase.
\medskip

\noindent \textbf{Subcase 2.2.} There exists a dangerous set disjoint from $Z$.

Let $X$ be an inclusionwise minimal dangerous set disjoint from $Z$.
\smallskip

\noindent \textbf{Subcase 2.2.1.} For any $t\in X$ and for any $s_1,s_2\in N_\gp(t)$, the set $\{s_1,s_2\}$ is feasible.

Take an arbitrary buyer $t_0\in T-X$ who shares a neighbor with $X$ and let $s_0\in N_\gp(t_0)\cap N_\gp(X)$. Let $G'$ denote the graph obtained by deleting $X\cup (N_\gp(X)-s_0)$. For any $(1,2)$-factor of $\gp$ containing $s_0t_0$, its restriction to $G'$ is a $(1,2)$-factor as well. In $G'$, some of the edges might not be contained in any of the $(1,2)$-factors. Still, by induction and Lemma~\ref{lem:fix}, there exists an adequate ordering $\sigma'$ of the items in $G'$. Let $\sigma''$ be an arbitrary ordering of the items in $N_\gp(X)-s_0$, and define $\sigma\coloneqq (\sigma',\sigma'')$. Then $t_0$ chooses at most one item from $N_\gp(X)$ (namely $s_0$), since she has at least one neighbor outside of $N_\gp(X)$ and those items have smaller indices in the ordering. Thus the adequateness of $\sigma$ follows from that of $\sigma'$ and from the assumption that any pair $s_1,s_2\in N_\gp(t)$ form a feasible set for $t\in X$.

\smallskip

\begin{figure}[t!]
\centering
\begin{subfigure}[t]{0.4\textwidth}
  \centering
  \includegraphics[width=.78\linewidth]{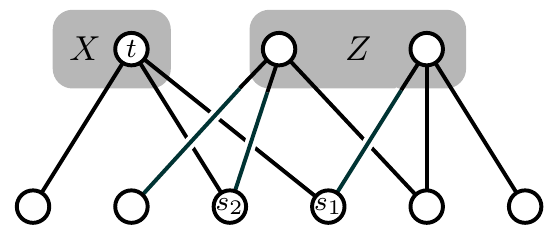}
  \caption{The graph of tight edges corresponding to the instance on Figure~\ref{fig:ex}, where $Z$ is an inclusionwise maximal dangerous set, and $X$ is an inclusionwise minimal dangerous set disjoint from $Z$.}
  \label{fig:alg1}
\end{subfigure}\hfill
\begin{subfigure}[t]{0.5\textwidth}
  \centering
  \includegraphics[width=\linewidth]{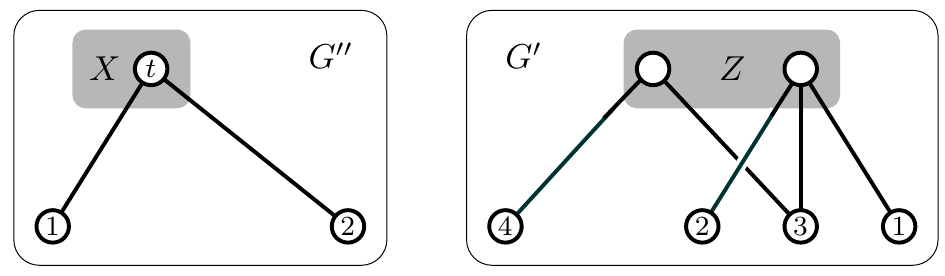}
  \caption{The graphs $G'=\gp-(X\cup (N_\gp(X)-s_1))$ and $G''=\gp-(Z\cup (N_\gp(Z)-s_1))$, together with an adequate ordering $\sigma'$ and an arbitrary ordering $\sigma''$, respectively.}
  \label{fig:alg2}
\end{subfigure}\vspace{5pt}
\begin{subfigure}[t]{0.4\textwidth}
  \centering
  \includegraphics[width=.78\linewidth]{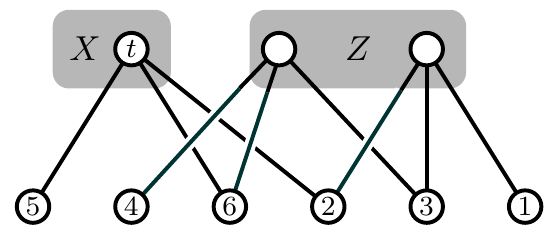}
  \caption{Construction of the ordering $\sigma=(\sigma',\sigma''|_{N_\gp(X)-s_2},\sigma'''$), where $\sigma'''$ is the trivial ordering of the one element set $\{s_2\}$.}
  \label{fig:alg4}
\end{subfigure}
\caption{An illustration of the inductive step in Subcase 2.2.2.}
\label{fig:alg}
\end{figure}

\noindent \textbf{Subcase 2.2.2.} There exists $t_0\in X$ and $s_1,s_2\in N_\gp(t)$ such that $\{s_1,s_2\}$ is not feasible.

The following claim is the key observation of the proof.

\begin{claim} \label{cl:key2}
$X\cup Z=T$ and $N_\gp(X)\cap N_\gp(Z)=\{s_1,s_2\}$. 
\end{claim}
\begin{proof}
Let $Y\subseteq T-t_0$ be a dangerous set with $s_1,s_2\in N_\gp(t_0)$. As $t_0\in T-(Z\cup Y)$ and $Z$ is inclusionwise maximal, either $Y\subseteq Z$ or $Y\cap Z=\emptyset$ by Claim~\ref{cl:dang}\ref{eq:cap}. In the latter case, $X$ and $Y$ are dangerous sets with $X\cup Y\subsetneq T$. Furthermore, $|N_\gp(X) \cap N_\gp(Y)| \geq 2$ since $s_1$ and $s_2$ are contained in both. Hence, by Claim~\ref{cl:dang}\ref{eq:cup}, $X\cap Y\neq\emptyset$. But then $X\cap Y$ is dangerous by Claim~\ref{cl:dang}\ref{eq:cap}, contradicting the minimality of $X$. Therefore, we have $Y\subseteq Z$. By Claim~\ref{cl:dang}\ref{eq:cup}, $X\cup Z=T$. As $|N_\gp(X)|=2|X|+1$, $|N_\gp(Z)=2|Z|+1$, and $|S|=2|T|=2|T|+2|Z|$, the claim follows.
\end{proof}

Let $G'$ and $G''$ denote the graphs obtained by deleting $X\cup(N_\gp(X)-s_1)$ and $Z\cup (N_\gp(Z)-s_1)$, respectively, see Figure~\ref{fig:alg}. For any $(1,2)$-factor of $\gp$ containing $t_0s_2$, its restriction to $G'$ is a $(1,2)$-factor as well. In $G'$, some of the edges might not be contained in any of the $(1,2)$-factors. Still, by induction and Lemma~\ref{lem:fix}, there exists an adequate ordering $\sigma'$ of the items in $G'$. Let $\sigma''$ be an arbitrary ordering of the items in $N_\gp(X)-s_2$. Finally, let $\sigma'''$ denote the trivial ordering of the single element set $\{s_2\}$. Let $\sigma\coloneqq (\sigma',\sigma''|_{N_\gp(X)-s_1},\sigma''')$. We claim that $\sigma$ is adequate. Indeed, if a buyer $t\in Z$ arrives first, then she chooses two items from $N_\gp(Z)-s_2$ according to $\sigma'$. As $\sigma'$ is adequate for $G'$ and $G''-s_1+s_2$ has a $(1,2)$-factor, the remaining graph has a $(1,2)$-factor as well. If a buyer $t\in X$ arrives first, then she chooses two items from $N_\gp(X)-s_2$ that form a feasible set, since the only pair that might not be feasible for her is $\{s_1,s_2\}$ by Claim~\ref{cl:key2}. 
\medskip

\noindent \textbf{Case 3.} $|N_\gp(T')|= 2|T'|$ for some $\emptyset\neq T'\subsetneq T$.

We claim that there exists a set $T'$ satisfying the assumption if and only if $G_\pi$ is not connected. Indeed, if $G_\pi$ is not connected, then necessarily the number of items is exactly twice the number of buyers in every component as the graph is supposed to have a $(1,2)$-factor. To see the other direction, let $S'\coloneqq N_\gp(T')$, $T''\coloneqq T-T'$, $S''\coloneqq S-S'$, and consider the subgraphs $G'\coloneqq G_\pi[T'\cup S']$ and $G''\coloneqq G_\pi[T''\cup S'']$. As every tight edge is legal and all the vertices in $S'$ are matched to vertices in $T'$ in any optimum $b$-matching, $G_\pi$ contains no edges between $T''$ and $S'$. Therefore, $G_\pi$ is not connected, and it is the union of $G'$ and $G''$. By induction, there exist adequate orderings $\sigma'$ and $\sigma''$ of $S'$ and $S''$, respectively. Then the ordering $\sigma\coloneqq (\sigma',\sigma'')$ is adequate with respect to $\pi$. 

\medskip

\begin{algorithm}[h!]
  \caption{Determining an adequate ordering for bi-demand markets with property~\eqref{eq:FAC}.}\label{alg:bi}
  \begin{algorithmic}[1]
    \Statex \textbf{Input:} Graph $G_\pi$ of tight edges, upper bounds $b(t)=2$ for $t\in T$ and $b(s)=1$ for $s\in S$.
    \Statex \textbf{Output:} Adequate ordering $\sigma$ of the items. 
	\If{$|N_\gp(Y)|\geq 2|Y|+2$ for every $\emptyset\neq Y\subsetneq T$} \label{st:case1}
		\State Let $\sigma$ be an arbitrary ordering of $S$.
	\ElsIf{$|N_\gp(Y)|\geq 2|Y|+1$ for every $\emptyset\neq Y\subsetneq T$, and there exists $Y$ for which equality holds} \label{st:case2}
	    \State Determine an inclusionwise maximal dangerous set $Z$. \label{st:max}
	    \If{there exists no dangerous set disjoint from $Z$} \label{st:disj}
	        \State Take an item $s_0\in N_\gp(Z)$ that has a neighbor $t_0\in T-Z$.
	        \State Let $\sigma'$ be an arbitrary ordering of $S-N_\gp(Z)$.
	        \State Determine an adequate ordering $\sigma''$ for $G''\coloneqq \gp[Z\cup (N_\gp(Z)-s_0)]$.
	        \State Let $\sigma'''$ be the trivial ordering of the single item $s_0$.
	        \State Set $\sigma\coloneqq (\sigma',\sigma'',\sigma''')$.
	    \Else
	        \State Determine an inclusionwise minimal dangerous set $X$ disjoint from $Z$. \label{st:min}
	        \If{$\{s_1,s_2\}$ is feasible for any $t\in X$ and $s_1,s_2\in N_\gp(t)$} \label{st:feas}
	            \State Take an item $s_0\in N_\gp(X)$ that has a neighbor $t_0\in T-X$.
	            \State Determine an adequate ordering $\sigma'$ for $G'\coloneqq \gp-(X\cup (N_\gp(X)-s_0))$.
	            \State Let $\sigma''$ be an arbitrary ordering of $N_\gp(X)-s_0$.
	            \State Set $\sigma\coloneqq (\sigma',\sigma'')$.
	        \Else{~~(Observation: $X\cup Z=T$ and $N_\gp(X)\cap N_\gp(Z)=\{s_1,s_2\}$.)}
	            \State Determine an adequate ordering $\sigma'$ for $G'\coloneqq \gp-(X\cup (N_\gp(X)-s_1))$.
	            \State Let $\sigma''$ be an arbitrary ordering of the items in $G''\coloneqq \gp-(Z\cup(N_\gp(Z)-s_1))$.
	            \State Let $\sigma'''$ be the trivial ordering of the single item $s_2$.
	            \State Set $\sigma\coloneqq (\sigma',\sigma''|_{N_\gp(X)-s_1},\sigma''')$.
	        \EndIf
        \EndIf
    \Else{~~(Observation: the graph $\gp$ is not connected.)} \label{st:case3}
    \State Let $\emptyset\neq T'\subsetneq T$ be a set with $|N_\gp(T')|=2|T'|$.
        \State Determine an adequate ordering $\sigma'$ for $G'\coloneqq \gp[T'\cup N_\gp(T')]$.
        \State Determine an adequate ordering $\sigma''$ for in $G''\coloneqq \gp-(T'\cup N_\gp(T'))$.
        \State Set $\sigma\coloneqq (\sigma',\sigma'')$.
	\EndIf
    \State \textbf{return} $\sigma$ \label{st:end}
  \end{algorithmic}
\end{algorithm}

By Lemma~\ref{lem:dual}, $\pi$ can be determined in polynomial time, hence the graph of tight edges is available. The algorithm for determining an adequate ordering for $G_\pi$ is presented as Algorithm~\ref{alg:bi}. To see that all steps can be performed in polynomial time, it suffices to show how to decide whether a pair $\{s_1,s_2\}$ of items forms a feasible set for a buyer $t$, and how to find an inclusionwise maximal or minimal dangerous set, if exists, efficiently. Checking the feasibility of $\{s_1,s_2\}$ for $t$ reduces to finding a $(1,2)$-factor in $\gp-\{s_1,s_2,t\}$. Dangerous sets can be found as follows: take two copies of each vertex $t\in T$, and connect them to the vertices in $N_\gp(t)$. Furthermore, add a dummy vertex $w_0$ to the graph and connect it to every vertex in $S$. Let $G'=(S',T';E')$ denote the graph thus obtained. For a set $Y\subseteq T$, let $Y'\subseteq T'$ consist of the copies of the vertices in $Y$ plus the vertex $w_0$. It is not difficult to check that $Y\subseteq T$ is an inclusionwise minimal or maximal dangerous set of $G_\pi$ if and only if $Y'$ is an inclusionwise minimal or maximal subset of $T'$ with $|N_{G'}(Y')|=|Y'|$. Hence $Y$ can be determined, for example, by relying on K\H{o}nig's alternating path algorithm \cite{konig1916graphen}. When an inclusionwise minimal dangerous set $X$ is needed that is disjoint from $Z$, then the same approach can be applied for the graph $\gp-Z$.  
\end{proof}

\begin{rem}
Theorem~\ref{thm:bi} settles the existence of optimal dynamic prices when the demand of each buyer is exactly two. However, the proof can be straightforwardly extended to the case when the demand of each buyer is at most two. 
\end{rem}

\section{Conclusions and open problems}

We considered the existence of optimal dynamic prices for multi-demand valuations. By relying on structural properties of an optimal dual solution, we gave polynomial-time algorithms for determining such prices in unit-demand markets and in multi-demand markets up to three buyers under a technical assumption, thus giving new interpretations of results of Cohen-Addad et al. and Berger et al. We also proved that any bi-demand market satisfying the same technical assumption has a dynamic pricing scheme that achieves optimal social welfare. 

It remains an interesting open question whether an analogous approach works when the total demand of the buyers exceeds the number of available items. Another open problem is to decide the existence of optimal dynamic prices in multi-demand markets in general.

\medskip

\paragraph{Acknowledgements}

The authors are truly grateful to Joseph Cheriyan, David Kalichman and Kanstantsin Pashkovich for pointing out a flaw in an earlier version of the paper, which motivated the addition of Lemma~\ref{lem:fix}. 

The work was supported by the Lend\"ulet Programme of the Hungarian Academy of Sciences -- grant number LP2021-1/2021 and by the Hungarian National Research, Development and Innovation Office -- NKFIH, grant numbers FK128673 and TKP2020-NKA-06.

\bibliographystyle{abbrv}
\bibliography{dynamic}

\end{document}